\documentclass[11pt,reqno]{amsart}
\usepackage{amsmath,amsthm,amssymb,bm,bbm,dsfont}
\usepackage{mathtools}\mathtoolsset{centercolon}
\usepackage[shortlabels]{enumitem}
\usepackage[mathscr]{eucal}
\usepackage{amsaddr}
\usepackage{upgreek}
\usepackage[left=2cm,right=2cm,top=2cm,bottom=2cm]{geometry}
\mathtoolsset{showonlyrefs=true}
\usepackage{setspace}
\usepackage{graphicx}
\usepackage{verbatim}
\usepackage{float}
\usepackage{placeins}
\usepackage{array, soul}
\usepackage{booktabs}
\usepackage{threeparttable}
\usepackage[update,prepend]{epstopdf}
\usepackage{multirow}
\usepackage{amsfonts,amssymb,dsfont}
\usepackage[abs]{overpic}
\usepackage{xfrac}
\usepackage{enumitem}

\usepackage[usenames,dvipsnames]{color}
\usepackage{colortbl}
\usepackage[hidelinks]{hyperref}
\hypersetup{
	unicode=false,          
	pdftoolbar=true,        
	pdfmenubar=true,        
	pdffitwindow=false,     
	pdfstartview={FitH},    
	pdftitle={My title},    
	pdfauthor={Author},     
	pdfsubject={Subject},   
	pdfcreator={Creator},   
	pdfproducer={Producer}, 
	pdfkeywords={keyword1} {key2} {key3}, 
	pdfnewwindow=true,      
	colorlinks=true,        
	linkcolor=Red,          
	citecolor=ForestGreen,  
	filecolor=Magenta,      
	urlcolor=BlueViolet,    
}
\usepackage{doi}
\usepackage{url}
\usepackage{caption, subcaption}
\usepackage{enumitem}
\usepackage{shadethm}

\makeatletter
\ifx\@NODS\undefined%

\let\mathbb=\mathds
\else%
\fi
\makeatother

\DeclareMathOperator{\Tr}{Tr}

\DeclareMathOperator{\spec}{spec}
\DeclareMathOperator{\e}{\mathrm{e}}

\newcommand{\be}{{\mathbf e}}

\def\0{{\mathbf{0}}}
\def\1{{\mathbf{1}}}
\def\2{{\mathbf{2}}}
\def\3{{\mathbf{3}}}
\def\4{{\mathbf{4}}}
\def\5{{\mathbf{5}}}
\def\6{{\mathbf{6}}}

\def\7{{\mathbf{7}}}
\def\8{{\mathbf{8}}}
\def\9{{\mathbf{9}}}

\def\be{\begin{equation}}
	\def\ee{\end{equation}}
\def\bea{\begin{eqnarray}}
	\def\eea{\end{eqnarray}}

\theoremstyle{plain}
\newshadetheorem{theo}{Theorem} 
\newtheorem*{theorem_BLP}{Theorem BLP}
\newtheorem*{theorem_GBLP}{Theorem GBLP}
\newtheorem{conj}{Conjecture} 
\newshadetheorem{prop}[theo]{Proposition} 
\newtheorem{lemm}[theo]{Lemma} 

\theoremstyle{definition}

\theoremstyle{remark}
\newtheorem{remark}{Remark}[section]

\newcommand{\vertiii}[1]{{\left\vert\kern-0.25ex\left\vert\kern-0.25ex\left\vert #1
		\right\vert\kern-0.25ex\right\vert\kern-0.25ex\right\vert}}

\makeatletter
\newcommand{\opnorm}{\@ifstar\@opnorms\@opnorm}
\newcommand{\@opnorms}[1]{%
	\left|\mkern-1.5mu\left|\mkern-1.5mu\left|
	#1
	\right|\mkern-1.5mu\right|\mkern-1.5mu\right|
}
\newcommand{\@opnorm}[2][]{%
	\mathopen{#1|\mkern-1.5mu#1|\mkern-1.5mu#1|}
	#2
	\mathclose{#1|\mkern-1.5mu#1|\mkern-1.5mu#1|}
}
\makeatother

\begin{document}
	
	\let\origmaketitle\maketitle
	\def\maketitle{
		\begingroup
		\def\uppercasenonmath##1{} 
		\let\MakeUppercase\relax 
		\origmaketitle
		\endgroup
	}
	
	\title{\bfseries \Large{On Araki-Type Trace Inequalities
	}}
	
	\author{ \normalsize 
		{Po-Chieh Liu}$^{1,2}$
		and
		{Hao-Chung Cheng}$^{1\text{--}5}$
	}
	\address{\small  	
		$^1$Department of Electrical Engineering and Graduate Institute of Communication Engineering,\\ National Taiwan University, Taipei 106, Taiwan (R.O.C.)\\
		$^2$Department of Mathematics, National Taiwan University\\
		$^3$Center for Quantum Science and Engineering,  National Taiwan University\\
		$^4$Hon Hai (Foxconn) Quantum Computing Center, New Taipei City 236, Taiwan (R.O.C.)\\
		$^5$Physics Division, National Center for Theoretical Sciences, Taipei 10617, Taiwan (R.O.C.)\\
	}
	
	\email{\href{mailto:haochung.ch@gmail.com}{haochung.ch@gmail.com}}
	
	\date{\today}
	
	\begin{abstract} 
		In this paper, we prove a trace inequality
		$\Tr\left[ f(A) A^s B^s \right]
		\leq 
		\Tr\big[ f(A) (A^{1/2} B A^{1/2} )^s \big]$ for any positive and monotonically increasing function $f$, $s\in[0,1]$, and positive semi-definite matrices $A$ and $B$.
		On the other hand, if $s\in[0,1]$ and the map $x\mapsto x^sg(x)$ is positive and decreasing, then 
		$
		\Tr\big[ g(A) (A^{1/2} B A^{1/2} )^s \big]
		\leq 
		\Tr\left[ g(A) A^s B^s \right]$.
	\end{abstract}
	
	

	\maketitle
	
	
	\section{Introduction} \label{sec:introduction}
	
	The \emph{Golden--Thompson trace inequality} \cite{Gol65, Tho65, FT14}, i.e., for Hermitian matrices $H,K$,
	\begin{align} \label{eq:GT}
		\Tr\left[ \mathrm{e}^{H+K} \right]
		\leq \Tr\left[ \mathrm{e}^H \mathrm{e}^K \right],
	\end{align}
	and the \emph{Araki--Lieb--Thirring inequality} \cite{LT76, Ara90}: for $A,B\geq 0$, 
	\begin{align} \label{eq:ALT}
		\Tr\left[ \left( A^{1/2} B A^{1/2} \right)^r \right] 
		\leq \Tr\left[ A^{r/2} B^r A^{r/2} \right], \quad r\geq 1
	\end{align}
	(the inequality is reversed\footnote{In some literature, 
		\eqref{eq:ALT} for $r\geq 1$ is called Lieb--Thirring inequality \cite{LT76}, and the reversed inequality for $r\in[0,1]$ is sometimes called the Araki--Cordes inequality \cite{Ara90, Cordes87}. In this paper, we will use the term Araki--Lieb--Thirring inequality, referring to the inequality for all $r\geq 0$.}
	for $r \in [0,1]$)
	play important roles in statistical mechanics, quantum information theory, and quantum computing.
	For example, \eqref{eq:GT} implies that 
	\begin{align}
		D_{\alpha}(\rho\Vert\sigma):= \frac{1}{\alpha-1} \log \Tr\left[ \rho^{\alpha} \sigma^{1-\alpha} \right]
		\leq
		\frac{1}{\alpha-1} \log \Tr\left[ \e^{\alpha \log \rho + (1-\alpha) \log \sigma } \right] =: D_{\alpha}^{\flat}(\rho\Vert\sigma), \quad \forall \alpha \in (0,1),
	\end{align}
	where $D_{\alpha}(\rho\Vert\sigma)$ is the Petz--R\'enyi divergence \cite{Pet86}, $D_{\alpha}^{\flat}(\rho\Vert\sigma)$ is the log-Euclidean R\'enyi divergence, and $\rho,\sigma$ are density matrices.
	On the other hand, \eqref{eq:ALT} implies that
	\begin{align}
		D_{\alpha}^*(\rho\Vert\sigma) := \frac{1}{\alpha-1} \log \Tr\left[ \left( \sigma^{\frac{1-\alpha}{2\alpha}} \rho \sigma^{\frac{1-\alpha}{2\alpha}} \right)^{\alpha} \right]
		\leq D_{\alpha}(\rho\Vert\sigma), \quad \forall \alpha > 0,
	\end{align}
	where $D_{\alpha}^*(\rho\Vert\sigma)$ is the sandwiched R\'enyi divergence \cite{Mul13, WWY14}.
	Essentially, the Golden--Thompson inequality and the Araki--Lieb--Thirring inequality manifest the noncommutative nature of quantum mechanics and show that certain inequalities, instead of equalities, arise in analytical derivations. 
	
	Araki--Lieb--Thirring trace inequality \eqref{eq:ALT} can be strengthened to the log-majorization relation \cite{Ara90}:
	\begin{align} \label{eq:ALT_log}
		\left( A^{1/2} B A^{1/2} \right)^r 
		\prec_{\log}  A^{r/2} B^r A^{r/2} , \quad r\geq 1
	\end{align}
	(the log-majorization relation reverses for $r\in[0,1]$).
	Here, for matrices $A,B \in \mathcal{M}_n(\mathds{C})$ with nonnegative eigenvalues, the \emph{log-majorization} of $A$ by $B$, denoted by $A \prec_{\log} B$, means that
	\begin{align}
		\prod_{i=1}^k \lambda_i(A) \leq \prod_{i=1}^k \lambda_i(B), \quad k = 1,2\ldots, n,
	\end{align}
	and $\det{A} = \det{B}$, where $\lambda_i(A)$ is the $i$-th eigenvalue of $A$ arranged in non-increasing order.
	Likewise, \eqref{eq:GT} can be strengthened by \eqref{eq:ALT_log} and the Lie--Trotter formula \cite{AH94}:
	\begin{align}
		\e^{H+K} \prec_{\log} \left( \e^{pH/2} \e^{pK} \e^{pH/2} \right)^{1/p}, \quad \forall p>0.    
	\end{align}
	Ando and Hiai later obtained a log-majorization via Löwner--Heinz operator inequality and the \emph{antisymmetric tensor power trick} (see e.g., \cite[Lemma 2.4]{operator_book}), which complements the Golden--Thompson type
	\cite{AH94}:
	\begin{align} \label{eq:AH94}
		\left[ \mathrm{e}^{ pK/2 } \left( \mathrm{e}^{ -pK/2 } \mathrm{e}^{ pH} \mathrm{e}^{ -pK/2 } \right)^{\alpha} \mathrm{e}^{ pK/2 }
		\right]^{1/p}
		\prec_{\log}
		\e^{\alpha H + (1-\alpha) K}, \quad \alpha \in[0,1], \; p>0.
	\end{align}
	The result then implies that
	\begin{align}
		D_{\alpha}^{\flat} (\rho\Vert\sigma)
		\leq \frac{1}{\alpha-1} \log \Tr\left[\rho \left( \rho^{-1/2} \sigma \rho^{-1/2} \right)^{1-\alpha} \right] =: \widehat{D}_{\alpha}(\rho\Vert \sigma), \quad \alpha \in (0,1].
	\end{align}
	
	In addition, two relatively recent works further generalize these classical results: Bourin and Lee established a family of matrix inequalities from a two-variable functional,
	which in particular provides an alternative proof of the Golden--Thompson inequality in log-majorization (see \cite[Corollary 2.12]{BL16}), while Hiai developed a generalization of Araki’s log-majorization via a log-convexity theorem \cite[Theorem 3.1]{Hia16b}.

	Bebiano, Lemos, and Providência applied Ando and Hiai's proof technique, along with Furuta's operator inequality \cite{Fur87}, to obtain the following log-majorization of Araki's type.
	\begin{theorem_BLP}[BLP Inequality {\cite{BLP05}}] \hypertarget{theo:BLP}{} 
			For $A,B\geq 0$,
			\begin{align}
				A^{1+q} B^q \prec_{\log} A \left( A^{r/2} B^r A^{r/2} \right)^{q/r}, \quad 0<q\leq r.
			\end{align}
			Equivalently, 
			\begin{align} \label{eq:BLP}
				A^{t+s} B^s \prec_{\log} A^t \left( A^{1/2} B A^{1/2} \right)^s, \quad s\in(0,1], \; t>0.
			\end{align}
		\end{theorem_BLP}
		Theorem~\hyperlink{theo:BLP}{BLP} provides an alternative proof to the following trace inequality
		\begin{align} \label{eq:HP93}
			\Tr\left[ A\left( \log A + \log B \right) \right]
			\leq \frac{1}{p} \Tr\left[ A \log \left( A^{p/2} B^p A^{p/2} \right) \right], \quad p>0
		\end{align}
		(see also the early paper by Hiai and Petz \cite{HP93}, the complementation by Ando and Hiai \cite{AH94}, and the recent developments by Lieb and Carlen \cite{LCV17}).
		Inequality \eqref{eq:HP93} then implies that the following relation between Umegaki's quantum relative entropy $D(\rho\Vert\sigma)$ \cite{Ume56} and Belavkin--Staszewski's relative entropy $\widehat{D}(\rho\Vert \sigma)$ \cite{BS82}, i.e.,
		\begin{align}
			D(\rho\Vert\sigma) := \Tr\left[ \rho (\log \rho - \log \sigma ) \right]
			\leq 
			\Tr\left[ \rho \log \left( \rho^{1/2} \sigma^{-1} \rho^{1/2} \right) \right] =:
			\widehat{D}(\rho\Vert \sigma).
		\end{align}
		
		With efforts to develop the reversed BLP inequalities \cite{MNF08, Fur12, LS18}, the above log-majorizations can be unified as the following generalized BLP (GBLP) inequalities.\footnote{
			Some literature has studied different generalizations of the BLP inequality, \eqref{eq:BLP}, and used the term generalized BLP \cite{FNT07, MF09, Fur07, FMN22}.
			In this paper, we refer to the inequalities and the reversed versions with the possible range of parameters as in Theorem~\hyperlink{theo:GBLP}{GBLP} as the GBLP inequality.
		}
		\begin{theorem_GBLP}[GBLP Inequality {\cite{BLP05, MNF08, Fur12, LS18}}] \hypertarget{theo:GBLP}{}
			For $A,B\geq 0$,
			\begin{align} \label{eq:GBLP}
				A^{r+q} B^q \prec_{\log} A^r \left( A^{p/2} B^p A^{p/2} \right)^{q/p}, \quad  0< q \leq p, \; r\geq 0.
			\end{align}
			If either $0\leq r\leq p \leq q$ and $p>0$, or 
			$0\leq q < p$ and $-r \geq q$, then \eqref{eq:GBLP} holds with reversed log-majorization.
		\end{theorem_GBLP}
		
		Theorem~\hyperlink{theo:GBLP}{GBLP}, together with a result of Matharu and Aujla \cite{MA12} (which follows from Ando and Hiai’s work \cite[Corollary 2.3]{AH94}, or directly from the Furuta inequality \cite{Fur87}), yields the following relations:
		\begin{align}
			\begin{split} \label{eq:MA12}
				A\left( A^{-1/2} B A^{-1/2} \right)^{s} &\preceq_{\log} A^{1-s } B^{s}, \quad s \in [0,1] \cup [2,\infty);
				\\
				A\left( A^{-1/2} B A^{-1/2} \right)^{s} &\succeq_{\log} A^{1-s } B^{s}, \quad s \in [1,2].
			\end{split}
		\end{align}
		The relations given in \eqref{eq:MA12} are intriguing in the following sense.
		Considering $s\in [0,1]$, then Araki--Lieb--Thirring inequality shows that $\left( A^{-1/2} B A^{-1/2} \right)^{s} \succeq_{\log} A^{-s } B^{s}$.
		However, when multiplying both sides by $A$, distributing the power $s$ yields a reversed log-majorization $A\left( A^{-1/2} B A^{-1/2} \right)^{s} \preceq_{\log} A^{1-s } B^{s}$.
		Hence, the goal of this paper is to identify for which functions $f$ the following trace inequality holds,
		\begin{align} \label{eq:goal}
			\Tr\left[ f(A) A^s B^s \right]
			\leq 
			\Tr\left[ f(A) \left(A^{1/2} B A^{1/2}\right)^s \right], \quad \forall \, s\in[0,1],
		\end{align}
		or when the trace inequality reverses.
		Obviously, the BLP inequality \eqref{eq:BLP} already provides an answer to \eqref{eq:goal} for $f(A) = A^t$, $t>0$.
		However, whenever $f$ is not a power function, the core idea of proving Theorems~\hyperlink{theo:BLP}{BLP} and \hyperlink{theo:GBLP}{GBLP}, i.e., the antisymmetric tensor power trick with the Furuta inequality, is not applicable.
		
		Our main result is to show that any nonnegative and nondecreasing function $f$ satisfies \eqref{eq:goal} (see Theorem~\ref{theo:main_direct}), 
		and, for any $s\in [0,1]$ and any function $g$ such that $x\mapsto x^s g(x)$ is nonnegative and nonincreasing, \eqref{eq:goal} is reversed (see Proposition~\ref{prop:main_converse}).
		
		We provide our main results and the proofs in Section~\ref{sec:proof}.
		In Section~\ref{sec:conjecture}, we discuss the scenario for $s>1$.

		\section{An Araki-type trace inequality} \label{sec:proof}
		
		We first show that \eqref{eq:goal} holds if $f(A)$ is a specific type of orthogonal projection matrix.
		\begin{prop} \label{prop:order}
			Let $A = \sum_{i=1}^n \lambda_i P_i$ be the spectral decomposition of positive semi-definite $A$
			with $\lambda_1 \geq \lambda_2 \geq \cdots \geq \lambda_{n} \geq 0$, and let $\Pi_k := \sum_{i=1}^k P_i$.
			Then, for all $B\geq 0$,
			\[
			\Tr\left[ \Pi_k A^s B^s \right]
			\leq 
			\Tr\left[ \Pi_k \left(A^{1/2} B A^{1/2} \right)^s \right] , \quad \forall \, s\in[0,1],\, 1\leq k\leq n.
			\]	
		\end{prop}
		
		\begin{remark}
			Proposition~\ref{prop:order} holds for $\Pi_k$ being a projection onto the eigenspace associated with the top $k$ eigenvalues.
			The trace inequality may not be true once replacing $\Pi_k$ by an arbitrary eigenprojection $P_i$. 
			To see a counterexample, let $s=1/2$ and
			\begin{align}
				A = \begin{bmatrix} 2 & 0 \\ 0 & 1\end{bmatrix}, \,
				B = \begin{bmatrix} 1 & 1 \\ 1 & 1\end{bmatrix}, \,
				P_2 = \begin{bmatrix} 0 & 0 \\ 0 & 1\end{bmatrix}.
			\end{align}
			Then, $\Tr[ P_2 A^s B^s ] = 1/\sqrt{2} \not\leq \Tr[P_2(A^{1/2} B A^{1/2})^s ] \approx 0.5774$.
		\end{remark}
		
		In the following proofs and lemmas, we suppose $A>0$, and the case $A\geq 0$ then follows by taking the limit.
		To prove Proposition~\ref{prop:order}, we need the following two lemmas.
		
		\begin{lemm} \label{lemm:1}
			Following the notation in Proposition~\ref{prop:order}, 
			we denote $I - \Pi_k$ by $\Pi_k^\text{c}$, and let
			$\tilde{A}_k=\Pi_k A+\lambda_k \Pi_k^\textnormal{c} =  \sum_{i=1}^n \max\{\lambda_i,\lambda_k\} P_i$.
			Then, we have
			\[
			\Tr\left[ \Pi_k \left(\tilde{A}_k^{1/2} B \tilde{A}_k^{1/2}\right)^s \right]
			\leq 
			\Tr\left[ \Pi_k \left(A^{1/2} B A^{1/2} \right)^s \right], \quad \forall\, s \in [0,1].
			\]
			The inequality reverses for all $s\in[1,2]$.
		\end{lemm}
		
		\begin{proof}
			First, note that $A$ and each $\tilde{A}_k$ commute.
			Let $R=(A\tilde{A}_k^{-1})^{1/2}$, which is a Hermitian contraction, i.e.,~$R^2\leq I$. 
			For $s\in(0,1]$, since the function $x\mapsto x^s$ is operator concave on $[0,\infty)$ and vanishes at $x=0$, the operator Jensen's inequality \cite{HP82} implies that
			\[
			\left(A\tilde{A}_k^{-1}\right)^{1/2}\left(\tilde{A}_k^{1/2}B\tilde{A}_k^{1/2}\right)^s \left(A\tilde{A}_k^{-1}\right)^{1/2}
			= R\left(\tilde{A}_k^{1/2} B \tilde{A}_k^{1/2}\right)^s R
			\leq \left(R\tilde{A}_k^{1/2} B \tilde{A}_k^{1/2} R\right)^s
			= \left( A^{1/2} B A^{1/2}\right)^s,
			\]
			and the case $s=0$ is trivial.
			Therefore, we obtain
			\begin{align}
				\Tr\left[ \Pi_k \left(\tilde{A}_k^{1/2} B \tilde{A}_k^{1/2}\right)^s \right]
				&=\Tr\left[ \Pi_k \left(A\tilde{A}_k^{-1}\right)^{1/2}\left(\tilde{A}_k^{1/2}B\tilde{A}_k^{1/2}\right)^s \left(A\tilde{A}_k^{-1}\right)^{1/2} \right]\\
				&\leq 
				\Tr\left[ \Pi_k \left(A^{1/2} B A^{1/2}\right)^s \right],
			\end{align}
			using the fact that $\Pi_k(A\tilde{A}_k^{-1})^{1/2} = (A\tilde{A}_k^{-1})^{1/2}\Pi_k = \Pi_k$.
			
			The case $s\in[1,2]$ follows similarly, using instead the operator convexity of $x\mapsto x^s$, $s\in[1,2]$.
		\end{proof}
		
		Note that the BLP inequality implies that, for $A,B\geq 0$,
		\begin{align} \label{eq:BLP_trace}
			\Tr\left[ A^t A^s B^s \right]
			\leq 
			\Tr\left[ A^t \left(A^{1/2} B A^{1/2}\right)^s \right] , \quad \forall \, s\in[0,1],\, t\geq 0.
		\end{align}
		We have another lemma.
		\begin{lemm} \label{lemm:2}
			For $A>0$ and $B\geq 0$,
			\[
			\Tr\left[ A^t \left(A^{1/2} B A^{1/2}\right)^s \right]
			\leq 
			\Tr\left[ A^t A^s B^s \right] , \quad \forall \, s\in[0,1],\, t\leq -s.
			\]	
		\end{lemm}
		
		\begin{proof}
			Substitute $B$ by $A^{-1/2}BA^{-1/2}$ to \eqref{eq:BLP_trace}.
			Then
			\[
			\Tr\left[ A^t A^s \left(A^{-1/2} B A^{-1/2}\right)^s \right]
			\leq 
			\Tr\left[ A^t B^s \right], \quad \forall \, s\in[0,1],\, t\geq 0.
			\]
			Next, substitute $A$ by $A^{-1}$, then
			\[
			\Tr\left[ A^{-t-s} \left(A^{1/2} B A^{1/2}\right)^s \right]
			\leq 
			\Tr\left[ A^{-t} B^s \right], \quad \forall \, s\in[0,1],\; t\geq 0.
			\]
			Or, 
			\[
			\Tr\left[ A^t \left(A^{1/2} B A^{1/2}\right)^s \right]
			\leq 
			\Tr\left[ A^t A^s B^s \right] , \quad \forall \, s\in[0,1],\; t\leq -s.\qedhere
			\]
		\end{proof}
		
		\begin{proof}[Proof of Proposition~\ref{prop:order}]
			Combining the two different conditions regarding $t$ from \eqref{eq:BLP_trace} and Lemma~\ref{lemm:2}, we conclude that 
			for $f(x)=x^t$, $t\geq 0$ or $f(x)=-x^t$, $t\leq-s$, we have
			\begin{align} \label{eq:main_proof1}
				\Tr\left[ f(A) A^s B^s \right]
				\leq 
				\Tr\left[ f(A) \left(A^{1/2} B A^{1/2}\right)^s \right], \quad \forall \, s\in[0,1].
			\end{align}
			
			Notice that $\Tr[\Pi_kA^s B^s]=\Tr[\Pi_k\tilde{A}_k^s B^s]$ for $\Pi_k$ introduced in Lemma~\ref{lemm:1}.
			Hence, to complete the proof of Proposition~\ref{prop:order}, by Lemma~\ref{lemm:1} we only need to show that
			\begin{align} \label{eq:main_proof2}
				\Tr\left[ \Pi_k \tilde{A}_k^s B^s \right]
				\leq 
				\Tr\left[ \Pi_k \left(\tilde{A}_k^{1/2} B \tilde{A}_k^{1/2}\right)^s \right] , \quad \forall \, s\in[0,1].
			\end{align}
			To show this, we first introduce $\tilde{A}_{k,\varepsilon}=\Pi_k A+(\lambda_k -\varepsilon)\Pi_k^{\mathrm{c}} $ as an approximation of $\tilde{A}_k$, for $\varepsilon\in (0,\lambda_k)$, where $\lambda_k$ is the $k$-th largest eigenvalue of $A$.
			
			Define 
			\[
			g_t(x)=1-\left(\frac{x}{\lambda_k -\varepsilon}\right)^t,  \quad \forall  \, t\leq -s,
			\]
			then, \eqref{eq:main_proof1} implies that
			\[
			\Tr\left[ g_t(\Tilde{A}_{k,\varepsilon}) \Tilde{A}_{k,\varepsilon}^s B^s \right]
			\leq 
			\Tr\left[ g_t(\Tilde{A}_{k,\varepsilon}) \left(\tilde{A}_{k,\varepsilon}^{1/2} B \tilde{A}_{k,\varepsilon}^{1/2} \right)^s \right] , \quad \forall \, s\in[0,1],\, t\leq -s,
			\]
			noting that the constant function $1$ is just $f(x)=x^t$ at $t=0$.
			
			Let $t\to -\infty$, then $g_t(\tilde{A}_{k,\varepsilon})\to \Pi_k$, and hence \[
			\Tr\left[ \Pi_k \tilde{A}_{k,\varepsilon}^s B^s \right]
			\leq 
			\Tr\left[ \Pi_k \left(\tilde{A}_{k,\varepsilon}^{1/2} B \tilde{A}_{k,\varepsilon}^{1/2}\right)^s \right] , \quad \forall \, s\in[0,1].
			\]
			Finally, by letting $\varepsilon\to 0$, we obtain
			\[
			\Tr\left[ \Pi_k \Tilde{A}_{k}^s B^s \right]
			\leq 
			\Tr\left[ \Pi_k \left(\tilde{A}_{k}^{1/2} B \tilde{A}_{k}^{1/2} \right)^s \right] , \quad \forall \, s\in[0,1].
			\]
			
			In conclusion, for all $s\in [0,1]$, we have
			\[
			\Tr\left[ \Pi_k A^s B^s \right]
			=
			\Tr\left[ \Pi_k \tilde{A}_{k}^s B^s \right]
			\leq 
			\Tr\left[ \Pi_k \left(\tilde{A}_{k}^{1/2} B \tilde{A}_{k}^{1/2}\right)^s \right] 
			\leq
			\Tr\left[ \Pi_k \left(A^{1/2} B A^{1/2} \right)^s \right],
			\]
			which completes the proof of Proposition~\ref{prop:order}.
		\end{proof}
		
		Now we are in a position to prove our main results.
		
		\begin{theo} \label{theo:main_direct}
			Let $A,B\geq 0$.
			For any nonnegative and nondecreasing function $f$ on an interval $\mathcal{J}$, where $\spec(A)\subseteq \mathcal{J}$, we have 
			\[
			\Tr\left[ f(A) A^s B^s \right]
			\leq 
			\Tr\left[ f(A) \left(A^{1/2} B A^{1/2} \right)^s \right], \quad \forall \, s\in[0,1].
			\]
		\end{theo}
		\begin{proof}
			Since $f$ is nonnegative and nondecreasing, we can choose some appropriate nonnegative coefficients $c_k$ for each $k=1,\cdots,n$, such that $f(A)=\sum_{k=1}^n c_k\Pi_k$. Then, the claim follows from Proposition~\ref{prop:order}.
		\end{proof}

		As above, we may also apply substitutions to obtain another direction of the same inequality, but with a different condition.
		\begin{prop} \label{prop:main_converse}
			Let $A,B\geq 0$.
			For any $s\in[0,1]$ and any function $g$ on an interval $\mathcal{J}\subseteq [0,\infty)$, where $\spec(A)\subseteq \mathcal{J}$, such that $x\mapsto x^sg(x)$ is nonnegative and nonincreasing, then 
			\begin{align*}
				\Tr\left[ g(A) \left(A^{1/2} B A^{1/2} \right)^s \right]
				\leq 
				\Tr\left[ g(A) A^s B^s \right].
			\end{align*}
		\end{prop}
		
		\begin{proof}
			Again, we assume $A>0$, and extend the result to $A\geq 0$ via a limit.
			Set $f(x)=x^{-s}g(x^{-1})$, which is nonnegative and nondecreasing. Since the domain of $f$ includes the spectrum of $A^{-1}$, by Theorem~\ref{theo:main_direct},
			\begin{align}
				\Tr\left[g(A)\left(A^{{1}/{2}}BA^{{1}/{2}}\right)^s\right]
				&=
				\Tr\left[ f\left(A^{-1}\right) \left(A^{-1}\right)^s \left(A^{{1}/{2}}BA^{{1}/{2}}\right)^s \right]
				\\
				&\leq 
				\Tr\left[ f\left(A^{-1}\right) \left(A^{-1/2} \cdot A^{{1}/{2}}BA^{{1}/{2}}\cdot A^{-1/2}\right)^s \right]
				\\
				&=\Tr\left[g(A)A^sB^s\right].\qedhere
			\end{align}
		\end{proof}
		
		\section{A Conjecture for \(s > 1\)} \label{sec:conjecture}
		
		Theorem~\ref{theo:main_direct} shows the trace inequality for all $s\in[0,1]$.
		Inspired by the reverse BLP inequality \cite{MNF08, FNT07, MF09}, we conjecture the following reverse direction.
		
		\begin{conj} \label{conj:s>1}
			For any nonnegative and nondecreasing function $f$ on an interval $\mathcal{J}$, where $\spec(A)\subseteq \mathcal{J}$, we have 
			\[
			\Tr\left[ f(A) A^s B^s \right]
			\geq 
			\Tr\left[ f(A) \left(A^{1/2} B A^{1/2} \right)^s \right], \quad \forall \, s \geq 1.
			\]
		\end{conj}
		
		Note that $x\mapsto x^s$ is operator convex for $s\in[1,2]$. We may therefore adapt the previous technique to prove Conjecture~\ref{conj:s>1} for $s\in[1,2]$, provided that the condition $0\leq r\leq p \leq q$, $p>0$ in the reverse direction of Theorem~\hyperlink{theo:GBLP}{GBLP} can be relaxed to $0< p \leq q$, $r\ge 0$.
		Therefore, we conjecture the following reverse log-majorization, which complements \eqref{eq:GBLP}.
		\begin{conj} \label{conj:GBLP_reverse}
			For $A,B\geq 0$,
			\begin{align}
				A^{r+q} B^q \succ_{\log}  A^r \left( A^{p/2} B^p A^{p/2} \right)^{q/p} , \quad  0 < p \leq q, \; r\geq 0.
			\end{align}
		\end{conj}
		\noindent (To show Conjecture~\ref{conj:s>1} for $s\in[1,2]$, the trace inequality version of Conjecture~\ref{conj:GBLP_reverse} would be sufficient.)
		
		\medskip 
		
		In fact, we can directly prove Conjecture~\ref{conj:s>1} for a special case of $s=2$ without resorting to Conjecture~\ref{conj:GBLP_reverse}.
		\begin{prop} \label{prop:main_s=2}
			For any nonnegative and nondecreasing function $f$ on an interval $\mathcal{J}$, where $\spec(A)\subseteq \mathcal{J}$, we have 
			\[
			\Tr\left[ f(A) A^s B^s \right]
			\geq 
			\Tr\left[ f(A) \left(A^{1/2} B A^{1/2} \right)^s \right], \quad s = 2.
			\]
		\end{prop}
		
		\begin{proof}
			Our claim is based on Lemma~\ref{lemm:1} and the reverse direction to \eqref{eq:main_proof2}, which is presented in the following Lemma~\ref{lemm:order_s=2}.
		\end{proof}
		
		\begin{lemm} \label{lemm:order_s=2}
			Let $A = \sum_{i=1}^n \lambda_i P_i$ be the spectral decomposition of positive semi-definite $A$ with $\lambda_1 \geq \lambda_2 \geq \cdots \geq \lambda_{n} \geq 0$,
			$\Pi_k := \sum_{i=1}^k P_i$,
			$\Pi_k^{\textnormal{c}} := I - \Pi_k$,
			and $\tilde{A}_k := \Pi_k A + \lambda_k \Pi_k^{\textnormal{c}}$.
			Then,
			\begin{align}
				\Tr\left[ \Pi_k \left(\tilde{A}_k^{1/2} B \tilde{A}_k^{1/2}\right)^s \right]
				\leq \Tr\left[ \Pi_k \tilde{A}_k^s B^s \right], \quad s = 2.
			\end{align}
		\end{lemm}
		\begin{proof}
			Using the fact that $\lambda_k \Pi_k \leq \Pi_k \tilde{A}_k \Pi_k$ by construction and the Araki--Lieb--Thirring trace inequality \eqref{eq:ALT}, we calculate
			\begin{align}
				\Tr\left[ \Pi_k \left(\tilde{A}_k^{1/2} B \tilde{A}_k^{1/2}\right)^2 \right]
				&= \Tr\left[ \Pi_k \tilde{A}_k \Pi_k  B { \tilde{A}_k} B \right]
				\\
				&= \Tr\left[ \Pi_k \tilde{A}_k \Pi_k  B {  \left(\Pi_k \tilde{A}_k \Pi_k + \lambda_k \Pi_k^{\text{c}} \right) } B \right]
				\\
				&= \Tr\left[ \Pi_k \tilde{A}_k \Pi_k  B { \Pi_k \tilde{A}_k \Pi_k} B \right]
				+ \Tr\left[ \Pi_k \tilde{A}_k \Pi_k  B { (\lambda_k \Pi_k^{\text{c}})} B \right]
				\\
				&= \Tr\left[ (\Pi_k \tilde{A}_k \Pi_k  \Pi_k B \Pi_k )^2 \right]
				+ \Tr\left[ \Pi_k \tilde{A}_k \Pi_k  B { (\lambda_k \Pi_k^{\text{c}})} B \right]
				\\
				&= \Tr\left[ (\Pi_k \tilde{A}_k \Pi_k  \Pi_k B \Pi_k )^2 \right]
				+ \Tr\left[ \Pi_k \tilde{A}_k \Pi_k  { \lambda_k \Pi_k }  B  {  \Pi_k^{\text{c}}} B \right]
				\\
				&\leq\Tr\left[ (\Pi_k \tilde{A}_k \Pi_k  \Pi_k B \Pi_k )^2 \right]
				+ \Tr\left[ \Pi_k \tilde{A}_k \Pi_k  {  \Pi_k \tilde{A}_k \Pi_k }  B  {  \Pi_k^{\text{c}}} B \right]
				\\
				&= \Tr\left[ (\Pi_k \tilde{A}_k \Pi_k  \Pi_k B \Pi_k )^2 \right]
				+ \Tr\left[ \Pi_k \tilde{A}_k^2 \Pi_k  B  {  \Pi_k^{\text{c}}} B \right]
				\\
				&\overset{\eqref{eq:ALT}}{\leq} \Tr\left[ (\Pi_k \tilde{A}_k \Pi_k )^2 (\Pi_k B \Pi_k )^2 \right]
				+ \Tr\left[ \Pi_k \tilde{A}_k^2 \Pi_k  B  {  \Pi_k^{\text{c}}} B \right]
				\\
				&= \Tr\left[ \Pi_k \tilde{A}_k^2 \Pi_k B \Pi_k B \right]
				+ \Tr\left[ \Pi_k \tilde{A}_k^2 \Pi_k  B  {  \Pi_k^{\text{c}}} B \right]
				\\
				&= \Tr\left[ \Pi_k \tilde{A}_k^2 \Pi_k  B  (\Pi_k + {  \Pi_k^{\text{c}}}) B \right]
				\\
				&= \Tr\left[ \Pi_k \tilde{A}_k^2 B^2 \right].\qedhere
			\end{align}
		\end{proof}
		
		\section{Conclusions} \label{sec:conclusions}
		In this paper, we have derived Araki-type trace inequalities that involve general monotone functions, extending beyond simple powers. Although usual Araki-type inequalities
		and the BLP inequalities are expressed in a stronger form of log-majorization,
		the antisymmetric tensor power trick may not be applicable to this broader class of functions. On the other hand, trace inequalities have broad and direct applications in fields such as quantum information science and mathematical physics. This motivates our effort to extend Araki-type trace inequalities to a more general setting.
		
		Numerical experiments suggest that Conjecture~\ref{conj:s>1} for $s> 1$ and a slightly generalized reverse BLP inequality in Conjecture~\ref{conj:GBLP_reverse} could be true, which remains for future work.
		
		\section*{Acknowledgments}
		We are supported by the Emerging Young Scholars Program of the National Science and Technology Council, Taiwan (R.O.C.) under Grants No.~NSTC 114-2628-E-002 -006, NSTC 114-2119-M-001-002, and NSTC 114-2124-M-002-003, by the Yushan Young Scholar Program of the Ministry of Education, Taiwan (R.O.C.) under Grants No.~ NTU-114V2016-1 and by the research project `Forefront Quantum Computing, Learning, and Engineering in Noisy Intermediate-Scale Quantum Era’ of National Taiwan University under Grant NTU-114L895005. H.-C.~Cheng acknowledges the support from the `Center for Advanced Computing and Imaging in Biomedicine (NTU-114L900702)’ through The Featured Areas Research Center Program within the framework of the Higher Education Sprout Project by the Ministry of Education (MOE) in Taiwan.

		
		\bibliographystyle{ieeetr}
		\newcommand{\bibliographytypesize}{\large}
		\bibliography{reference, operator}

	\end{document}